\documentclass[a4paper,twocolumn,11pt,accepted=2017-05-09]{quantumarticle}
\pdfoutput=1
\usepackage[utf8]{inputenc}
\usepackage[english]{babel}
\usepackage[T1]{fontenc}
\usepackage{amsmath}
\usepackage{hyperref}

\usepackage{tikz}
\usepackage{lipsum}
\usepackage{amsthm}
\usepackage{color}
\usepackage{amssymb}

\usepackage{braket}
\usepackage[numbers,sort&compress]{natbib}
\usepackage{enumitem}
\usepackage{longtable}
\newcommand{\BH}{{\mathcal{B}(\mathcal{H})}}

\newcommand\map{S}
\newtheorem{thm}{Theorem}

\newtheorem{pro}{Proposition}

\newtheorem{rem}{Remark}
\newtheorem{ex}{Example}
\newcommand{\tr}{\operatorname{tr}}
\newcommand{\adjug}{\#}
\newcommand{\conj}{\dagger}

\begin{document}

\title{On the Alberti-Uhlmann Condition for Unital Channels}

\author{Sagnik Chakraborty}
\email{sagnik@fizyka.umk.pl}
\affiliation{Institute of Physics, Faculty of Physics, Astronomy and Informatics, Nicolaus Copernicus University,
Grudzi\k{a}dzka 5/7, 87-100 Toru\'n, Poland}

\author{Dariusz Chru{\'s}ci{\'n}ski}
\email{darch@fizyka.umk.pl}
\affiliation{Institute of Physics, Faculty of Physics, Astronomy and Informatics, Nicolaus Copernicus University,
Grudzi\k{a}dzka 5/7, 87-100 Toru\'n, Poland}

\author{Gniewomir Sarbicki}
\email{gniewko@fizyka.umk.pl}
\affiliation{Institute of Physics, Faculty of Physics, Astronomy and Informatics, Nicolaus Copernicus University,
Grudzi\k{a}dzka 5/7, 87-100 Toru\'n, Poland}

\author{Frederik vom Ende}
\email{frederik.vom-ende@tum.de}
\affiliation{Department of Chemistry, Technische Universit\"at M\"unchen, 85747 Garching, Germany}
\affiliation{Munich Centre for Quantum Science and Technology (MCQST), Schellingstr.~4,
80799 M\"unchen, Germany}

\maketitle

\begin{abstract}
  We address the problem of existence of completely positive trace preserving (CPTP) maps between two sets of density matrices. We refine the result of Alberti and Uhlmann and derive a necessary and sufficient condition for the existence of a unital channel between two pairs of qubit states which ultimately boils down to three simple inequalities.
\end{abstract}

\section{Introduction}

A fundamental aspect of resource theories is finding conditions which characterize the possibility of state-conversion via ``allowed'' operations. In quantum thermodynamics, for example, one usually asks whether a state $\tau$ can be generated from an initial state $\rho$ via a Gibbs-preserving channel, that is, a completely positive and trace-preserving (\textsc{cptp}) map which leaves the Gibbs-state of the system invariant \cite{horodecki13, brandao, gour, buscemi1, vE1}. Sometimes one limits the problem from Gibbs-preserving channels to the smaller set of thermal operations which are of a specific form compliant with the underlying resource theory. However, in what follows we focus on general \textsc{cptp} maps.  {Recall that complete positivity of a linear map $T : M_n(\mathbb{C})\to M_n(\mathbb{C})$ is equivalent to positivity of ${\rm id} \otimes T : M_n(M_n(\mathbb{C})) \to M_n(M_n(\mathbb{C}))$} \cite{Paulsen,Stormer}.

The above is just a special case of the following more general question: given two pairs of quantum states $\{\tau_1,\tau_2\}$, $\{\rho_1,\rho_2\}$ can one find a characterization (which is non-trivial and possibly simple to verify) for the existence of a quantum channel which transforms both $\rho_j$ into $\tau_j$? While this problem is fully solved in the classical case (for an overview on the equivalent conditions see \cite[Ch.~V]{gour} or \cite[Prop.~4.2]{vE2}), in the quantum realm this remains unanswered, with the qubit case being the notable exception: In a seminal paper \cite{alberti} Alberti and Uhlmann derived a necessary and sufficient condition for such a simultaneous state conversion. The Alberti-Uhlmann condition reads
\begin{equation}  \label{AU}
\Vert \rho_1-t\rho_2\Vert_1\geq\Vert \tau_1-t\tau_2\Vert_1~~~~~ \forall ~t\geq 0,
\end{equation}
where $\Vert A\Vert_1=\tr\sqrt{A^{\conj}A}$ denotes the trace-norm \citep{nielson}. Note that \eqref{AU} is equivalent to
\begin{equation}
\Vert p_1\rho_1-p_2\rho_2\Vert_1\geq\Vert p_1\tau_1-p_2\tau_2\Vert_1,
\label{AU1}
\end{equation}
for all $p_1,p_2\geq 0$ and $p_1+p_2=1$. Due to Helstrom theory \cite{Hel} condition (\ref{AU1}) means that distinguishability of states $\{\rho_1,\rho_2\}$ given with probabilities $\{p_1,p_2\}$ is not less than distinguishability of states $\{\tau_1,\tau_2\}$ (occurring with the same probabilities).
For the sake of computation \eqref{AU} can be reduced to finitely many inequalities via the formulae
\begin{align}
\|A\|_1^2&=\tr(A^\conj A)+2|\det(A)|\notag\\
\det(A+B)&=\det(A)+\det(B)+\tr(A^\adjug B)\label{eq:det_sum}
\end{align}
for all $A,B\in M_2(\mathbb C)$ where $(\cdot)^\adjug $ denotes the adjugate \citep[Ch.~0.8]{HJ1}
{ given by
\begin{equation*}
M^\adjug=\begin{bmatrix}
d & -b\\ -c & a
\end{bmatrix}, \text{where}~
M=\begin{bmatrix}
a & b\\ c & d
\end{bmatrix} .
\end{equation*}
Thus} the function resulting from \eqref{AU} is piecewise quadratic in $t$ so non-negativity reduces to certain conditions on the coefficients (with respect to $t$).

It is well known that every positive trace-preserving (\textsc{ptp}) map is a contraction in the trace-norm and hence the Alberti-Uhlmann condition (\ref{AU}) is necessary for the existence of a \textsc{ptp} map $T$. Interestingly,
this condition is sufficient in two dimensions --
even for the existence of a quantum channel --
and fails to be sufficient for dimension three and larger \citep[Ch.~VII.B]{heino}.

%{\color{red}\textsc{cptp} maps find major application in describing open quantum systems, in particular, in areas relating to Markovianity %\citep{breuer-book, BLP-review,rhp-review}.}

In \citep{chefles}, Chefles et al.~generalized the problem to input and output sets %\st{to be}
$\{\rho_1,\ldots,\rho_n\}$ and $\{\tau_1,\ldots,\tau_n\}$, respectively, with arbitrary dimension and arbitrary value of $n$, under the constraint that at least one of the two sets must be a set of pure states. They derived conditions for the existence of a \textsc{cptp} map between the sets in terms of the Gram matrices %\citep{chefles}
of the two sets.

A result for arbitrary (non-pure) states was derived by Huang et al.~\cite{poon} where their characterization (of existence of a \textsc{cptp} map) goes via the existence of some more abstract decomposition of the initial and target states. More interestingly, they considered the case of qubit states $\{\rho_1,\rho_2,\rho_3\}$, $\{\tau_1,\tau_2,\tau_3\}$ under the generic assumption of the input states being pure (cf.~footnote \footnote{
%Be aware that the case of pure input states in the qubit case is generic:
For arbitrary qubit states $\rho_1,\rho_2$ there exists $c\in[0,1)$ such that $\tilde\rho_1:=\frac{\rho_1-c\rho_2}{1-c}$ is a pure state. %and, due to \eqref{AU}, $\tilde\tau_1:=\frac{\tau_1-c\tau_2}{1-c}$ is still positive semi-definite
Thus the problem $\{\rho_1,\ldots,\rho_n\}\to\{\tau_1,\ldots,\tau_n\}$ by linearity of the desired channel can be reduced to $\{\tilde\rho_1,\ldots,\tilde\rho_n\}\to\{\tilde\tau_1,\ldots,\tilde\tau_n\}$.
}).
However the characterization they derived, while verifiable with standard software, seems to not generalize to a condition for arbitrary input states and, moreover, seems to not lead to much physical insight.

In \citep{heino}, Heinosaari et al.~considered a slightly different version of the problem and studied conditions for the existence of only \textsc{cp} maps between two sets of quantum states. Moreover they gave a fidelity characterization for the existence of a \textsc{cptp} transformation on pairs of qubit states, consisting of only a finite number of conditions. %Also they presented a counterexample to their characterization as well as to Alberti-Uhlmann for dimension three and larger.

In a more recent paper \citep{buscemi1}, Dall'Arno et al.~derived a condition for the existence of a \textsc{cptp} map when the input set is a collection of qubit states which can be, through a simultaneous unitary rotation, written as real matrices. They study the problem from the perspective of quantum statistical comparison and derive that if the testing region of the real input states includes the testing region of the output states then there exists a \textsc{cptp} map connecting them.
{ More precisely, the main result of \citep{buscemi1} says that a set of qubit states represented by real density matrices $\{\rho_1,\ldots,\rho_n\}$ can be mapped via a quantum channel to a set $\{\tau_1,\ldots,\tau_n\}$ if and only if the corresponding testing regions satisfy
\begin{equation}\label{RR}
  \mathcal{R}(\{ \tau_k\}) \subset  \mathcal{R}(\{ \rho_k\}) ,
\end{equation}
where $ \mathcal{R}(\{ \rho_k\}) $ is a set of vectors $ \mathbf{x} \in \mathbb{R}^n$ which can be realized via
$$   x_k = {\rm Tr}(\rho_k \pi) \ , \ \ \ k=1,\ldots, n , $$
for some $0\leq \pi \leq \openone$. Similarly one defines $ \mathcal{R}(\{ \tau_k\})$. } For further analysis on the relation of this problem to quantum statistical comparison see \citep{buscemi2,buscemi3}.

In this paper we refine the original Alberti-Uhlmann problem by asking about the existence of a unital channel, that is, $T$ maps $\rho_k$ to $\tau_k$ and additionally $T(\openone) = \openone$. Original condition (\ref{AU}) guarantees the existence of a \textsc{cptp} map but does say nothing  whether $T$ is unital. Clearly, condition (\ref{AU}) is again necessary but no longer sufficient. Note, that the map $T$ is uniquely defined only on an at most 3-dimensional subspace $\mathcal{M}$ spanned by $\{\openone,\rho_1,\rho_2\}$ and one asks whether this map can be extended to the whole algebra $M_2(\mathbb{C})$ such that the extended map is \textsc{cptp} and unital. Extension problems such as this one were already considered by many authors  before \cite{Paulsen}. The classical result of Arveson \cite{Arveson} says that if $\mathcal{M}$ is an operator system in $\BH$, that is, $\mathcal{M}$ is a linear subspace closed under hermitian conjugation and containing $\openone$, and if  $\Phi : \mathcal{M} \to \BH$ is unital and completely positive, then it can be extended to a unital completely positive map $\widetilde{\Phi} : \BH \to \BH$. Note, however, that this result says nothing about trace-preservation. Hence, even if the unital map $\Phi$ is trace-preserving the unital extension $\widetilde{\Phi}$ need not be trace-preserving. Actually, unitality may be relaxed by assuming that the hermitian subspace $\mathcal{M}$ contains a strictly  positive operator  \cite{heino}. Interestingly, it was shown \cite{Jencova} that if $\mathcal{M}$ is spanned by positive operators and $\Phi$ is completely positive then there exists a completely positive extension $\widetilde{\Phi}$.

The main result of this paper reads as follows. { Let $\{\rho_1,{\rho_2}\}$ and  $\{\tau_1,\tau_2\}$ be two pairs of qubit states and consider the two subspaces of $M_2(\mathbb{C})$
\begin{align*}
\mathcal{M}&:= {\rm span}\{{\rho_0=}\openone,\rho_1,\rho_2\}, \\ \mathcal{N}&:= {\rm span}\{{\tau_0=}\openone,\tau_1,\tau_2\}.
\end{align*}
%$$ \mathcal{M}:= {\rm span}\{{\rho_0=}\openone,\rho_1,\rho_2\} \ ; \ \ \ \mathcal{N}:= {\rm span}\{{\tau_0=}\openone,\tau_1,\tau_2\} . $$
In what follows we assume w.l.o.g.~that $\operatorname{dim} \mathcal{M} = %\operatorname{dim} \mathcal{N}= 
3$---indeed $\operatorname{dim} \mathcal{M} =1$ is trivial and $\operatorname{dim} \mathcal{M} =2$ can be traced back to the original result of Alberti and Uhlmann.}

\begin{thm}\label{mainthm}
The following statements are equivalent:
\begin{itemize}
\item[(i)] There exists a unital quantum channel mapping $\{\rho_1,{\rho_2}\}$ to $\{\tau_1,\tau_2\}$.\smallskip
\item[(ii)] {There exists a unital \textsc{ptp} map which sends $\{\rho_1,{\rho_2}\}$ to $\{\tau_1,\tau_2\}$.}\smallskip
\item[(iii)] For all $ \alpha, \beta, \gamma \in \mathbb{R}$
  \begin{eqnarray}\label{MAIN}
    && \lVert \alpha \openone + \beta \tau_1 + \gamma \tau_2 \rVert_1  \le \lVert \alpha \openone + \beta \rho_1 + \gamma \rho_2 \rVert_1\ \ \ \
  \end{eqnarray}
\item[(iv)] For all $\beta, \gamma \in \mathbb{R}$
  \begin{eqnarray}\label{MAIN_2}
    && \lVert  \tfrac{\openone}{2}+ \beta \tau_1 + \gamma \tau_2 \rVert_1  \le \lVert \tfrac{\openone}{2} + \beta \rho_1 + \gamma \rho_2 \rVert_1
  \end{eqnarray}
  \item[(v)] For all $t\in\mathbb R$ one has $\tau_1-t\tau_2 \prec \rho_1-t\rho_2$.
\end{itemize}
\end{thm}
Here $\prec$ denotes classical matrix majorization which is usually defined via the comparison of eigenvalues { but can be characterized as follows (e.g., \cite[Ch.~7]{ando}):
 for any two Hermitian $n \times n $ matrices one has $B  \prec  A$ iff there is a unital quantum channel $T$ such that $B = T(A)$. Interestingly    \cite[Thm.~7.4]{ando} $B  \prec  A$ iff for every real $t$ one has $\| \openone - t B\|_1 \leq \|\openone - t A\|_1$ which may be considered the order-free characterization of majorization.}

Clearly, the original Alberti-Uhlmann condition (\ref{AU}) provides now only the necessary condition corresponding to $\alpha=0$ in (\ref{MAIN}). %This condition readily serves as the necessary condition for existence of the unital channel as the trace-norm is contractive under the action of any \textsc{ptp} map. %\cite{Wolf06}.

While the conditions in Theorem \ref{mainthm} give conceptional insight, one can also reduce the problem { to} three easy-to-verify inequalities.
\begin{thm}\label{mainthm_2}
There exists a unital quantum channel mapping qubit states $\{\rho_1,\rho_2\}$ to $\{\tau_1,\tau_2\}$ if and only if  $ \det(\tau_j)\geq\det(\rho_j)$ for $j=1,2$ as well as
  \begin{equation}\label{MAIN_3}
    \begin{split}
     &\big( \tr( \rho_1^\adjug \rho_2 )-\tr( \tau_1^\adjug \tau_2 ) \big)^2\\
     &\leq 4\big(\det(\tau_1)-\det(\rho_1)\big)\big(\det(\tau_2)-\det(\rho_2)\big).
     \end{split}
  \end{equation}
  %where $(\cdot)^\adjug $ denotes the adjugate. \cite[Ch.~0.8]{HJ1},
%  \footnote{
%  Recall that in two dimensions one has the simple identities
%  $  \{\{a,b\},\{c,d\}\}^\adjug =\{\{d,-b\},\{-c,a\}\}$
%  as well as
%  $
%  \det(A+B)=\det(A)+\det(B)+\tr(A^\adjug  B)
%  $
%  for all $a,b,c,d\in\mathbb C$, $A,B\in M_2(\mathbb C)$.
%  }.
\end{thm}

\section{Proof of Alberti-Uhlmann conditions for unital maps}

%Let $\mathcal{M}:= {\rm span}\{{\color{red}\rho_0=}\openone,\rho_1,\rho_2\}$ and $\mathcal{N}:= {\rm span}\{{\color{red}\tau_0=}\openone,\tau_1,\tau_2\}$.

{
Consider a linear map $T : \mathcal{M} \to \mathcal{N}$ mapping $\rho_k$ to $\tau_k$.  By construction this map is unital and preserves the trace as well as hermiticity.
%{\color{red}Due to the following lemma we may assume (without loss of generality) $\operatorname{dim}(\mathcal M)=3$.
%\begin{lem}
%Theorem \ref{mainthm} holds for all $\{\rho_1,\rho_2\}$, $\{\tau_1,\tau_2\}$ if and only if it holds for only those initial states which satisfy $\operatorname{dim}(\mathcal{M})=3$.
%\end{lem}
%\noindent The proof idea goes as follows: the case $\operatorname{dim}(\mathcal M)=1$ is trivial and if $\operatorname{dim}(\mathcal M)=2$ then this can be traced back to the original Alberti-Uhlmann condition.\medskip
%
%.......
%}

\begin{pro} \label{prop_unitary_channel} The map $T$ is unitarily equivalent to a Pauli diagonal map $S : \Sigma \to \Sigma$ where, here and henceforth, $\Sigma = \operatorname{span}\{\openone,\sigma_x,\sigma_y\}$.
\end{pro}
}
\begin{proof} Since $\mathcal{M}$ is a 3-dimensional operator system in $M_2(\mathbb{C})$ there exists a {hermitian} operator orthogonal to $\mathcal{M}$ w.r.t.~the Hilbert-Schmidt inner product. Being orthogonal to $\openone$ it is traceless so after choosing an appropriate basis it equals $\sigma_z$.
Analogously, choosing an appropriate basis on the output Hilbert space makes $\mathcal{N}$ orthogonal to $\sigma_z$. More precisely this yields unitary $\tilde U,\tilde V\in M_2(\mathbb C)$  such that
$$   \tilde{U}^\conj \, \mathcal{M} \, \tilde{U} \perp \sigma_z \ ; \ \ \  \tilde{V}^\conj \, \mathcal{N} \, \tilde{V} \perp \sigma_z \,. $$
Now, define the map $\tilde\map(\rho):=\tilde V^\conj T(\tilde U\rho \tilde U^\conj )\tilde V$. Since $\tilde\map$  preserves hermiticity and the trace one has
$$
  \tilde\map(\sigma_{x}) = \left[ \begin{array}{cc} 0 & z_x \\  {z^*_x}  & 0 \end{array} \right], \qquad   \tilde\map(\sigma_{y}) = \left[ \begin{array}{cc} 0 & z_y \\ {z^*_y}& 0 \end{array} \right] ,
$$
with $z_x,z_y \in \mathbb{C}$. Hence the action of $\tilde\map$ can be represented via
%acts on $(\openone,\sigma_x,\sigma_y)\in (M_2(\mathbb C))^3$ as
$$
%\frac12\big(\langle\sigma_j,V^\conj T(U\sigma_jU^\conj )V\rangle\big)_{j=0,x,y}=
\begin{bmatrix} \openone \\ \sigma_x \\\sigma_y   \end{bmatrix}\ \to \
\begin{bmatrix} 1&0&0\\0&a_{11}&a_{12}\\0&a_{21}&a_{22} \end{bmatrix} \begin{bmatrix} \openone \\ \sigma_x \\\sigma_y   \end{bmatrix} ,
$$
for some $A=(a_{jk})_{j,k=1}^2\in\mathbb R^{2\times 2}$. Applying the singular value decomposition to $A$ one finds orthogonal matrices $W_1,W_2\in M_2(\mathbb R)$ such that $W_1^TAW_2=\operatorname{diag}(a,b)$ with $a,b\geq 0$ \cite[Thm.~7.3.5]{HJ1}. But every orthogonal $2\times 2$ matrix is a rotation matrix (possibly up to a composition with $\sigma_z$) \cite[p.~68]{HJ1}. {Hence,} the channel $S_\phi(\rho):=U_\phi^\conj \rho U_\phi$ with $U_\phi=\operatorname{diag}(1,e^{i\phi})$, $\phi\in\mathbb R$  leaves $\operatorname{span}\{\openone,\sigma_x,\sigma_y\}$ invariant and acts on $(\openone,\sigma_x,\sigma_y)$ as
$$ \begin{bmatrix} \openone \\ \sigma_x \\\sigma_y   \end{bmatrix}\ \to \
\begin{bmatrix} 1&0&0\\0&\cos\phi&-\sin\phi\\0&\sin\phi&\cos\phi \end{bmatrix} \begin{bmatrix} \openone \\ \sigma_x \\\sigma_y   \end{bmatrix}\ \,.
$$
Now, there exist $\phi,\theta\in\mathbb R$ such that
$
S_\phi\circ\tilde S\circ S_{-\theta}$ corresponds to the diagonal matrix $\operatorname{diag}(1,a,b)
$ or $\operatorname{diag}(1,a,-b)$. Thus one has
\begin{equation}\label{}
  S(\rho) =  V^\conj T(U\rho  U^\conj ) V ,
\end{equation}
where $U=\tilde UU_\theta$, $V=\tilde VU_\phi$, that is, the map $T$ is unitarily equivalent to the Pauli diagonal map $S$ defined on $\Sigma$.
\end{proof}

{
Therefore it is clear that  the unital map $T : \mathcal{M} \to \mathcal{N}$ can be extended to a unital \textsc{cp} map $\mathbf{T} : M_2(\mathbb{C}) \to M_2(\mathbb{C})$ if and only if the unital Pauli diagonal map $S : \Sigma \to \Sigma$ can be extended  to a unital \textsc{cp} map $\mathbf{S} : M_2(\mathbb{C}) \to M_2(\mathbb{C})$. Now, $T : \mathcal{M} \to \mathcal{N}$ is a positive map if and only if $S : \Sigma \to \Sigma$ is positive. However, for Pauli diagonal maps the positivity conditions are well-known \cite{Ruskai}: $S$ is positive  if and only if $|a|,|b|\leq 1$.
%
%\begin{cor}
%The map $T : \mathcal{M} \to \mathcal{N}$ mapping $\rho_k$ to $\tau_k$ is positive if and only if the matrix
%\begin{equation}\label{}
%  a_{ij} = \frac 12 {\rm Tr}(\sigma_i T(\sigma_j)) \ ; \ \ \ i,j=x,y ,
%\end{equation}
%satisfies {\color{blue}$\| A \| \leq 1$ with $\| \cdot \|$ being the operator norm, that is, the largest singular value}.
%\end{cor}
Now, if $S : \Sigma \to \Sigma$ is positive it can be always extended to a unital \textsc{ptp} map $\mathbf{S}$ on $M_2(\mathbb{C})$. 
Indeed, one may use the trivial extension defining
\begin{equation}\label{}
  \mathbf{S}(\sigma_z) = 0 .
\end{equation}
\begin{rem}  Actually it was shown \cite{qubit} that if $\mathcal{M}$ is a 3-dimensional operator system, then there exists a unital \textsc{ptp} projector  $\Pi : M_2(\mathbb{C}) \to \mathcal{M}$. Hence $\mathbf{S} := S \circ \Pi$ defines a unital \textsc{ptp} extension of $S$. Interestingly, however, it was proven \cite{qubit} that there is no \textsc{cptp} projector  $\Pi : M_2(\mathbb{C}) \to \mathcal{M}$.
\end{rem}
Yet one can find an extension of $S$ which is not only \textsc{ptp} but even \textsc{cptp}:
\begin{pro} \label{EXT} If a Pauli diagonal map $S : \Sigma \to \Sigma$ is positive it can always be extended to unital \textsc{cptp} map $\mathbf{S}$ on $M_2(\mathbb{C})$.
\end{pro}
\begin{proof}
We show that there exists $\mathbf{S}$ which is also a Pauli diagonal map, that is,
$$  \mathbf{S}(\sigma_z) = c \sigma_z . $$
The Pauli diagonal map $\mathbf{S}$ is \textsc{cptp} if and only if the well-known Fujiwara-Algoet conditions are satisfied \cite{Algoet,Ruskai}
\begin{equation}\label{FA}
  |a+b| \leq 1 + c \ , \ \ |a-b|\leq 1 - c\,,
\end{equation}
that is,
\begin{equation*}
  |a+b| - 1 \leq  c \leq 1 - |a-b| \,.
\end{equation*}
Clearly,  there exists a non-trivial solution for $c$ if and only if
  \begin{equation*}
    2 \ge |a-b| + |a+b| = 2 \max\{|a|,|b|\} ,
  \end{equation*}
which concludes the proof.
\end{proof}

It turns out that the above Proposition also has a simple geometric representation:
\begin{rem} It is well known \cite{Ruskai,Ruskai2,Karol} that the parameters $(a,b,c)$ of the Pauli diagonal unital \textsc{cptp} maps form a tetrahedron defined by Fujiwara-Algoet conditions (\ref{FA}) with vertices $\mathbf{x}_0=(1,1,1)$, $\mathbf{x}_1=(1,-1,-1)$, $\mathbf{x}_2=(-1,1,-1)$ and $\mathbf{x}_3=(-1,-1,1)$. These four vertices correspond to four unitary Pauli channels --- $\mathcal{U}_\alpha(\rho) = \sigma_\alpha \rho \sigma_\alpha$, with $\sigma_0 = \openone$, and $\{1,2,3\}=\{x,y,z\}$. Projecting the tetrahedron along the $z$-axis onto the $xy$-plane one obtains the unit square of parameters $a,b$ satisfying $|a|,|b|\leq 1$, cf.~Figure \ref{fig}. Hence it is clear that starting from any point in the unit square and going along the $z$-axis one eventually ends up inside the tetrahedron, thus realizing a \textsc{cptp} extension of a \textsc{ptp} map from the unit square.
\end{rem}

}

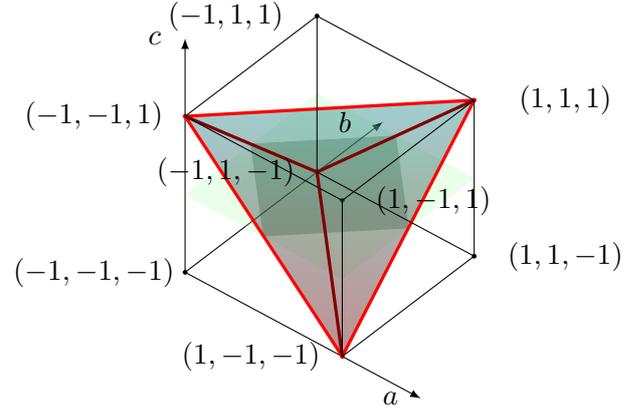
\begin{figure}
  \begin{tikzpicture}[scale=0.45]%[font=\LARGE]
    \def \alpha	{ 40 }
    \def \beta 	{ -40 }
    \def \a{ 6. }

    \coordinate (A) at (0,0);
    \coordinate (B) at ({	\a*cos(\alpha)			},{	\a*sin(\alpha)*sin(\beta)				});
    \coordinate (C) at ({	\a*sin(\alpha)			},{	-\a*cos(\alpha)*sin(\beta)				});
    \coordinate (D) at ({	0				},{	\a*cos(\beta)						});
    \coordinate (E) at ({	\a*(cos(\alpha) + sin(\alpha))	},{	\a*(sin(\alpha)-cos(\alpha))*sin(\beta)			});
    \coordinate (F) at ({	\a*sin(\alpha)			},{	-\a*cos(\alpha)*sin(\beta)+\a*cos(\beta)		});
    \coordinate (G) at ({	\a*cos(\alpha)			},{	\a*sin(\alpha)*sin(\beta)+\a*cos(\beta)			});
    \coordinate (H) at ({	\a*cos(\alpha)+\a*sin(\alpha)	},{	\a*((sin(\alpha)-cos(\alpha))*sin(\beta)+cos(\beta))	});

    \coordinate (BD) at ({ \a*cos(\alpha)/2			},{	\a*(sin(\alpha)*sin(\beta) + cos(\beta))/2 		});
    \coordinate (BH) at ({ \a*(2*cos(\alpha)+sin(\alpha))/2	},{	\a*((2*sin(\alpha)-cos(\alpha))*sin(\beta)+cos(\beta))/2});
    \coordinate (CD) at ({ \a*sin(\alpha)/2			},{	\a*(-cos(\alpha)*sin(\beta)+cos(\beta))/2		});
    \coordinate (CH) at ({ \a*(cos(\alpha)+2*sin(\alpha))/2	},{	\a*((sin(\alpha)-2*cos(\alpha))*sin(\beta)+cos(\beta))/2});
    \coordinate (A0) at ({ 0					},{	\a*cos(\beta)/2						});
    \coordinate (B0) at ({	\a*cos(\alpha)			},{	\a*sin(\alpha)*sin(\beta)+\a*cos(\beta)/2		});
    \coordinate (C0) at ({	\a*sin(\alpha)			},{	-\a*cos(\alpha)*sin(\beta)+\a*cos(\beta)/2		});
    \coordinate (E0) at ({	\a*(cos(\alpha) + sin(\alpha))	},{	\a*(sin(\alpha)-cos(\alpha))*sin(\beta)+\a*cos(\beta)/2	});

    \fill[black]  	(A) circle [radius=2pt];
    \fill[red]    	(B) circle [radius=2pt];
    \fill[red]		(C) circle [radius=2pt];
    \fill[red]    	(D) circle [radius=2pt];
    \fill[black]   	(E) circle [radius=2pt];
    \fill[black] 	(F) circle [radius=2pt];
    \fill[black]   	(G) circle [radius=2pt];
    \fill[red] 		(H) circle [radius=2pt];

%     \fill[green] 	(BD) circle [radius=2pt];
%     \fill[green] 	(BH) circle [radius=2pt];
%     \fill[green] 	(CD) circle [radius=2pt];
%     \fill[green] 	(CH) circle [radius=2pt];
%     \fill[green] 	(A0) circle [radius=2pt];
%     \fill[green] 	(B0) circle [radius=2pt];
%     \fill[green] 	(C0) circle [radius=2pt];
%     \fill[green] 	(E0) circle [radius=2pt];

    \draw[->,>=latex] (0,0) -- ({	1.5*\a*cos(\alpha)	},{	1.5*\a*sin(\alpha)*sin(\beta)		});
    \draw[->,>=latex] (0,0) -- ({	1.5*\a*sin(\alpha)	},{	-1.5	*\a*cos(\alpha)*sin(\beta)	});
    \draw[->,>=latex] (0,0) -- ({	0			},{	1.5*\a*cos(\beta)			});

%     \draw [-,thin] (B)  --  (A)
% 		   (Ap) -- (Bp)
% 		   (B)  --  (C)
% 		   (D)  --  (C)
% 		   (A)  --  (D)
% 		   (Ap) --  (A)
% 		   (Cp) --  (C)
% 		   (Bp) --  (B)
% 		   (Bp) -- (Cp);
%
%     % Draw the hidden edges of the parallelepiped
%     \draw [gray,-,thin] (Dp) -- (Cp);
% 			(Dp) --  (D);
% 			(Ap) -- (Dp);
%
%     % Name the vertices (the names are not consistent
%     %  with the node name, but it makes the programming easier)
%     \draw (Ap) node [right]           {$A$}
% 	  (Bp) node [right, gray]     {$F$}
% 	  (Cp) node [right]           {$D$}
% 	  (C)  node [left,gray]       {$E$}
% 	  (D)  node [left]            {$B$}
% 	  (A)  node [left,gray]       {$G$}
% 	  (B)  node [above left=+5pt] {$C$}
% 	  (Dp) node [right,gray]      {$H$};
%
%     % Drawing again vertex $C$, node (B) because it disappeared behind the edges.
%     % Drawing again vertex $H$, node (Dp) because it disappeared behind the edges.
%     \fill[red]   (B) circle [radius=2pt];
%     \fill[gray] (Dp) circle [radius=2pt];
%
%     % From the reference and this example one can easily draw
%     % the twin tetrahedron jointly to this one.
%     % Drawing the edges of the twin tetrahedron
%     % switching the p_s: A <-> Ap, etc...
%     \draw[red,-,dashed, thin] (A)  -- (Dp)
% 			      (A)  -- (Bp)
% 			      (A)  --  (C)
% 			      (Bp) -- (Dp)
% 			      (C)  -- (Dp)
% 			      (Bp) --  (C);

    \filldraw[draw=red,color=green!35!, 	opacity=.25] (A0) -- (B0) -- (E0) -- (C0);
    \filldraw[draw=red,color=green!25!black, 	opacity=.25] (BD) -- (CD) -- (CH) -- (BH);
    \filldraw[draw=red,bottom color=red!75!black, top color=cyan!75, opacity=.25]
      (B) -- (D)  -- (H);

    \draw[red,-,very thick] 		(B) -- (D)
					(D) -- (H)
					(H) -- (B);
    \draw[red!50!black,-,very thick] 	(C) -- (B)
					(C) -- (D)
					(C) -- (H);
    \draw 				(G) -- (B)
					(G) -- (D)
					(G) -- (H)
					(F) -- (C)
					(F) -- (D)
					(F) -- (H)
					(E) -- (C)
					(E) -- (B)
					(E) -- (H);

    \node[shift={(-1.2,0)}] at (A) {$(-1,-1,-1)$};
    \node[shift={(-1.2,0)}] at (B) {$(1,-1,-1)$};
    \node[shift={(-1.2,0)}] at (C) {$(-1,1,-1)$};
    \node[shift={(-1.2,0)}] at (D) {$(-1,-1,1)$};
    \node[shift={(1.2,0)}] at (E) {$(1,1,-1)$};
    \node[shift={(-1.2,0)}] at (F) {$(-1,1,1)$};
    \node[shift={(1.2,0)}] at (G) {$(1,-1,1)$};
    \node[shift={(1.2,0)}] at (H) {$(1,1,1)$};

    \node[shift={(-.4,0)}] at ({	1.5*\a*cos(\alpha)	},{	1.5*\a*sin(\alpha)*sin(\beta)		}) {$a$};
    \node[shift={(-.5,0)}] at ({	1.5*\a*sin(\alpha)	},{	-1.5	*\a*cos(\alpha)*sin(\beta)	}) {$b$};
    \node[shift={(-.4,0)}] at ({	0			},{	1.5*\a*cos(\beta)			}) {$c$};
\end{tikzpicture}

\caption{The tetrahedron represents Pauli diagonal channels with $a,b$ and $c$ being the respective eigenvalues for eigenvectors $\sigma_x,\sigma_y$ and $\sigma_z$. The horizontal plane represents Pauli maps with $c=0$.}\label{fig}
\end{figure}

% {\color{red}
%\begin{cor}A completely positive extension exists if and only if a positive extension exists.
%\end{cor}
%}

Now, we are ready to prove the main theorem.
\begin{proof}[Proof of Theorem~\ref{mainthm}]
{ (i) $\Rightarrow$ (ii) is obvious and (ii) $\Rightarrow$ (iii) follows from the well-known fact that the trace-norm is contractive under the action of \textsc{ptp} maps.}
% (iii) $\Rightarrow$ (ii): Multiply \eqref{MAIN_2} by $2|\alpha|$ for arbitrary $\alpha\in\mathbb R$ and observe that the transformation
% $
%(\alpha,\beta,\gamma) \mapsto (\alpha,2\alpha\beta,2\alpha\gamma)
% $
%on $\mathbb R^3$ with codomain $F:=\mathbb R^3\setminus\big(\{0\}\times(\mathbb R^2\setminus\{0\})\big)$ is bijective. Thus \eqref{MAIN} holds for all $(\alpha,\beta,\gamma)\in F$ and by a standard continuity argument extends to the closure $\overline{F}=\mathbb R^3$.
%
To show  (iii) $\Rightarrow$ (i) let us observe that the very condition (\ref{MAIN}) is equivalent to
 \begin{align} \label{MAIN-1}
 \lVert \alpha \openone + \beta a \sigma_x &+ \gamma b \sigma_y \rVert_1 \nonumber\\
 &\le \lVert \alpha \openone + \beta \sigma_x + \gamma \sigma_y \rVert_1
  \end{align}
for all $\alpha, \beta, \gamma \in \mathbb{R}$ {
with $a,b$ being the parameters of the Pauli diagonal channel $S$ as above}. Indeed, assume (\ref{MAIN}), i.e.~for all $\alpha,\beta,\gamma$%\in\mathbb R$
 \begin{align*}
 \|\alpha\openone+\beta V^\conj \tau_1V+\gamma &V^\conj \tau_2V\|_1=\|\alpha\openone+\beta \tau_1+\gamma \tau_2\|_1\\
&  \leq \|\alpha \openone+ \beta \rho_1+ \gamma\rho_2\|_1\\
&= \|\alpha \openone+ \beta U^\conj \rho_1U + \gamma U^\conj \rho_2U\|_1
%%%%%%%%%%%%%%%%
%  &=\|S(\alpha \openone+ \beta U^\conj \rho_1U  +\gamma U^\conj \rho_2U)\|_1\\
%  &=\|T(\alpha \openone+ \beta \rho_1+ \gamma\rho_2)\|_1\\
%&  \leq \|\alpha \openone+ \beta \rho_1+ \gamma\rho_2\|_1\\
%&= \|\alpha \openone+ \beta U^\conj \rho_1U  \gamma U^\conj \rho_2U\|_1
 \end{align*}
with $U,V$ being the unitary matrices from Proposition \ref{prop_unitary_channel}, using unitary equivalence of the trace norm.
Writing
 $$
 U^\conj \rho_jU=\frac{\openone}{2}+r_{jx}\sigma_x+r_{jy}\sigma_y
 $$
 for some $r_{jx},r_{jy}\in\mathbb R$ yields
 $$
 V^\conj \tau_jV=\frac{\openone}{2}+ar_{jx}\sigma_x+br_{jy}\sigma_y\,.
 $$
 Thus the trace norm condition (\ref{MAIN})  becomes:
\begin{align} \label{MAIN-2}
\lVert \alpha' \openone + \beta' a \sigma_x &+ \gamma' b \sigma_y \rVert_1\nonumber\\
& \le \lVert \alpha' \openone + \beta' \sigma_x + \gamma' \sigma_y \rVert_1 ,
\end{align} 
 
% \begin{eqnarray} \label{MAIN-2}
% \lVert \alpha' \openone + \beta' a \sigma_x + \gamma' b \sigma_y \rVert_1 \le \lVert \alpha' \openone + \beta' \sigma_x + \gamma' \sigma_y \rVert_1 ,
%  \end{eqnarray}
where
$$
\begin{bmatrix} \alpha' \\ \beta' \\ \gamma' \end{bmatrix} = \begin{bmatrix} 1&\frac12&\frac12\\0&r_{1x}&r_{2x}\\0&r_{1y}&r_{2y} \end{bmatrix} \begin{bmatrix} \alpha \\ \beta \\ \gamma \end{bmatrix} .
$$
{Now, observe that}
% \begin{align*}
% \|(\alpha+\tfrac{\beta}{2}+\tfrac{\gamma}{2})\openone+(r_{1x}\beta+r_{2x}\gamma)a\sigma_x+(r_{1y}\beta+r_{2y}\gamma)b\sigma_y \|_1\\
%\leq \|(\alpha+\tfrac{\beta}{2}+\tfrac{\gamma}{2})\openone+(r_{1x}\beta+r_{2x}\gamma)\sigma_x+(r_{1y}\beta+r_{2y}\gamma)\sigma_y \|_1
% \end{align*}
 $$
 \det\begin{bmatrix} 1&\frac12&\frac12\\0&r_{1x}&r_{2x}\\0&r_{1y}&r_{2y} \end{bmatrix}=\det\begin{bmatrix} r_{1x}&r_{2x}\\r_{1y}&r_{2y} \end{bmatrix}\neq 0
 $$
because $\operatorname{dim}(\mathcal M)=3$. Hence condition (\ref{MAIN}) for all $\alpha,\beta,\gamma$ is equivalent to (\ref{MAIN-2}) for all $\alpha',\beta',\gamma'$, that is, it is equivalent to  (\ref{MAIN-1}) for all $\alpha,\beta,\gamma$.

Now, {as we have shown (iii) implies (\ref{MAIN-1}), the last step is to show that (\ref{MAIN-1}) implies (i)}. Choosing $(\alpha,\beta,\gamma)=(0,1,0)$ and $(\alpha,\beta,\gamma)=(0,0,1)$ in (\ref{MAIN-1}) implies $|a|\leq 1$ and $|b|\leq 1$ so there exists a \textsc{cptp} extension of $S$ (and hence also of $T$) due to Proposition \ref{EXT}.

(iii) $\Leftrightarrow$ (iv) is readily verified using homogeneity and continuity of the norm.

(v) $\Leftrightarrow$ (iv): By \cite[Thm.~7.1]{ando} and \eqref{AU} condition (v) is equivalent to $\|\tfrac{\openone}{2}-s(\tau_1-t\tau_2)\|_1\leq\|\tfrac{\openone}{2}-s(\rho_1-t\rho_2)\|_1$ for all $s,t\in\mathbb R$ which by the same argument as before is equivalent to (iv).
\end{proof}

%{\color{blue}
%Condition (iii) from Theorem \ref{mainthm} can be geometrically interpreted as the linear subspace $\operatorname{span}\{\tau_1,\tau_2\}$ being to be closer to the maximally mixed state $\frac{\openone}{2}$ than $\operatorname{span}\{\rho_1,\rho_2\}$.
%}

%With the first theorem proven our second main result now is an easy consequence.

{

\begin{rem}
  Note that an input set $\{ \rho_0 = \openone/2,\rho_1,\rho_2\}$ can be always represented by real matrices and hence it fits the scenario of  \cite{buscemi1}. Therefore, condition (\ref{RR}) is equivalent to all conditions from Theorem 1.
\end{rem}

}

\begin{proof}[Proof of Theorem~\ref{mainthm_2}]
The eigenvalue formula for $2\times 2$ matrices directly yields that if $A,B\in M_2(\mathbb C)$ are hermitian and of same trace then $A\prec B$ iff $\det(A)\geq \det(B)$. Using \eqref{eq:det_sum}, condition (v) from Theorem \ref{mainthm} is equivalent to
\begin{align*}
t^2\big(\det(\tau_2)-\det(\rho_2)\big)&-t\big( \tr( \rho_1^\adjug \rho_2 )-\tr( \tau_1^\adjug \tau_2 )  \big)\\
& +\big(\det(\tau_1)-\det(\rho_1)\big)\geq 0
\end{align*}
for all $t\in\mathbb R$. But a parabola $t\mapsto at^2+bt+c$ is non-negative iff $a,c\geq 0$ and $b^2\leq 4ac$ which concludes the proof.
\end{proof}

\section{Example}{
Let us now present an example which hopefully demonstrates the power of our results. Consider the linear map
\begin{equation}\label{MAP-EX}
 \left[ \begin{array}{cc} \rho_{11} & \rho_{12} \\ \rho_{21}  & \rho_{22} \end{array} \right]   \mapsto 
\left[ \begin{array}{cc} \rho_{11} +(1-p)\, \rho_{22} & \kappa \rho_{12} \\ \kappa \rho_{21}  & p\, \rho_{22} \end{array} \right] ,
\end{equation}
which is \textsc{cptp} whenever $1\geq p \geq \kappa^2$. Choose $p=\kappa=\frac 12$ and recall that
every qubit state can be represented as
$$
 \left[ \begin{array}{cc} c&\gamma\sqrt{c(1-c)}\\\gamma^*\sqrt{c(1-c)}&1-c \end{array} \right]
$$
for some $c\in[0,1]$ and some $\gamma\in\mathbb C$, $|\gamma|\leq 1$. Thus we may consider the two valid density matrices:
\begin{align*}
  \rho_1 &= \left[ \begin{array}{cc} c & \frac{3i}{4}\sqrt{c(1-c)}\\-\frac{3i}{4}\sqrt{c(1-c)}&1-c \end{array} \right], \quad \\
  \rho_2 &= \left[ \begin{array}{cc} 0.2 & 0.4 \\ 0.4 & 0.8 \end{array} \right]
\end{align*}
with $c\in[0,1]$ arbitrary which are mapped to
\begin{align*}
  \tau_1 &= \left[ \begin{array}{cc} \frac12(1+c) & \frac{3i}{8}\sqrt{c(1-c)} \\ -\frac{3i}{8}\sqrt{c(1-c)} & \frac12(1-c) \end{array} \right], \quad \\
  \tau_2 &= \left[ \begin{array}{cc} 0.6 & 0.2 \\ 0.2 & 0.4 \end{array} \right] 
\end{align*}
by means of \eqref{MAP-EX}. The Alberti-Uhlmann condition is obviously satisfied since by constriction the $\tau_k$ are related to the $\rho_k$ via a \textsc{cptp} map.
Now let us use Theorem \ref{mainthm_2} to check for which $c$ one can find a unital extension: 
\begin{itemize}
\item $\det(\tau_1)\geq\det(\rho_1)$: one readily computes that this is satisfied iff $(1-c)(\frac14-\frac{21}{64}c)\geq 0$. Together with the base assumption $c\in[0,1]$ this leads to $c\in[0,\frac{16}{21}]$.
\item $\det(\tau_2)=0.2\geq 0=\det(\rho_2)$ $\checkmark$
\item Condition \eqref{MAIN_3} after a straightforward computation becomes
$
(0.7c-0.3)^2\leq 0.8(1-c)(\frac14-\frac{21}{64}c)
$ which together with $c\in[0,1]$ means $c\leq 0.6082\ldots$
\end{itemize}
While this fully solves the one-parameter problem let us look at two interesting and hopefully illuminating special cases.
\begin{ex}
Choosing $c=0$ we know there exists a unital channel which simultaneously maps $\{\rho_1=|1\rangle\langle 1|, \rho_2\}$ to $\{\tau_1=\openone/2$, $\tau_2\}$ with $\rho_2,\tau_2$ from above. 
%While neither Thm.~\ref{mainthm} nor Thm.~\ref{mainthm_2} provides an explicit unital channel which does this one quickly finds 
Interestingly, it can also be found that a unital \textsc{cptp} map giving rise to this transformation is given by
\begin{widetext}
\begin{align*}
 \left[ \begin{array}{cc} \rho_{11} & \rho_{12} \\ \rho_{21}  & \rho_{22} \end{array} \right] \, \mapsto \,
\left[ \begin{array}{cc}  \frac{\rho_{11}+\rho_{22}}{2}+\frac{\rho_{12}+\rho_{21}}{8}&\frac{\rho_{12}+\rho_{21}}{4}\\\frac{\rho_{12}+\rho_{21}}{4}&\frac{\rho_{11}+\rho_{12}}{2}-\frac{\rho_{12}+\rho_{21}}{8}  \end{array} \right].
\end{align*}
\end{widetext}
%which, among others, achieves $\{\rho_1,\rho_2\}\mapsto\{\tau_1,\tau_2\}$.
\end{ex}
\begin{ex}
For the second example choose $c=\frac23$ so by the above analysis no unital \textsc{cptp} map simultaneously transforms
\begin{equation*}
  \rho_1 = \left[ \begin{array}{cc} \frac23 & \frac{i\sqrt{2}}{8}\\-\frac{i\sqrt{2}}{8}&\frac13 \end{array} \right], \qquad
  \rho_2 = \left[ \begin{array}{cc} 0.2 & 0.4 \\ 0.4 & 0.8 \end{array} \right]
\end{equation*}
into
\begin{equation*}
  \tau_1 = \left[ \begin{array}{cc} \frac56&\frac{i\sqrt{2}}{16}\\-\frac{i\sqrt{2}}{16}&\frac16 \end{array} \right], \qquad
  \tau_2 = \left[ \begin{array}{cc} 0.6 & 0.2 \\ 0.2 & 0.4 \end{array} \right] \,.
\end{equation*}
However, $\tau_1\prec\rho_1$ and $\tau_2\prec\rho_2$ so by \cite[Thm.~7.1]{ando} one can find unital channels $T_1,T_2$ such that $T_j(\rho_j)=\tau_j$ for $j=1,2$.
This is what makes this example remarkable: while there exists a \textsc{cptp} map sending $\{\rho_1,\rho_2\}$ to $\{\tau_1,\tau_2\}$ (i.e.~\eqref{MAP-EX} for $p=\kappa=\frac12$) and while there exist unital channels which map either $\rho_1$ to $\tau_1$ or $\rho_2$ to $\tau_2$, there is no \textbf{simultaneous} unital \textsc{cptp} transformation. Thus this example clearly emphasizes the necessity of Theorem \ref{mainthm}.
\end{ex}

}

\section{Conclusions \& Outlook}

In this paper we derived necessary and sufficient conditions for the existence of a unital quantum channel mapping a pair of qubit states $\{\rho_1,\rho_2\}$ to $\{\tau_1,\tau_2\}$. These conditions connect the problem to trace-norm inequalities (in the spirit of Alberti-Uhlmann \eqref{AU} which is reproduced by setting $\alpha =0$ in \eqref{MAIN-1}) and majorization on matrices. Moreover we reduced the infinite family of conditions to just three inequalities which are simple enough to be verified with pen and paper. We also provided an example of two pairs of qubit states $\tau_k\prec\rho_k$ which satisfy the Alberti-Uhlmann condition, that is, there exists a quantum channel mapping $\rho_k$ to $\tau_k$, but condition (\ref{MAIN_3}) is violated meaning there is no single unital channel between $\rho_k$ and $\tau_k$. { The authors of \cite{heino} provide an example of a qutrit system (i.e.~3-dimensional) where pairs of states $\rho_1,\rho_2$ and $\tau_1,\tau_2$ satisfy the Alberti-Uhlmann condition (\ref{AU}) but there is no quantum channel mapping $\rho_k \to \tau_k$ (cf. \cite{heino}, Sec. VII B). Actually, the same example may be used to show that $\{\openone/2,\rho_1,\rho_2\}$ and $\{\openone/2,\tau_1,\tau_2\} $ satisfy (\ref{MAIN}) but there is no (unital) quantum channel mapping $\rho_k \to \tau_k$.    }

We expect that our result will encourage more research in this direction and shed light on finding more general closed form conditions for the existence of channels between sets of quantum states.
Possible next steps could focus on the case of the input set consisting of any three linearly independent qubit states or---in spirit of thermo- and general $D$-majorization \cite{vE3}---on how to modify Theorem \ref{mainthm} \& \ref{mainthm_2} if the fixed point of the channel is not the identity but an arbitrary Gibbs state (i.e.~an arbitrary positive-definite state $D$).

\section*{Acknowledgements}
%S.C., D.C.~and G.S.~
The authors are grateful to Teiko Heinosaari, Antonella De Pasquale, Namit Anand and Francesco Buscemi for fruitful discussion {as well as the anonymous referees'} constructive comments.
S.C., D.C.~and G.S.~were supported by the Polish National Science
Centre project 2018/30/A/ST2/00837. F.v.E.~is supported by the Bavarian excellence network {\sc enb} via the \mbox{International} PhD Programme of Excellence {\em Exploring Quantum Matter} ({\sc exqm}).

%\begin{ex} Consider a quantum channel $\Phi : \mathrm{L}(\mathcal{H}_2) \to \mathrm{L}(\mathcal{H}_2)$ defined by
%
%\begin{equation}\label{}
%  \Phi(X) = (1-p) X + p |2\rangle \langle 2| {\rm Tr}X ,
%\end{equation}
%with $p \in(0,1)$. Let $\rho_1 = |1\rangle \langle 1|$ and $\rho_2 = |2\rangle \langle 2|$. One has
%
%\begin{eqnarray*}
%% \nonumber to remove numbering (before each equation)
%  \tau_1 &=& \Phi(\rho_1) = |1\rangle \langle 1| \\
%  \tau_2 &=& \Phi(\rho_2) = p |1\rangle \langle 1| + (1-p) |2\rangle \langle 2| .
%\end{eqnarray*}
%By construction $\{\rho_1,\rho_2\}$ and $\{\tau_1,\tau_2\}$ satisfy (\ref{AU}). Condition (\ref{MAIN}) leads to
%
%\begin{equation}\label{}
%  |\alpha + \gamma (1-p)| + |\alpha+\beta+ p\gamma| \leq|\alpha+\gamma| + |\alpha+\beta| ,
%\end{equation}
%for all real $\alpha,\beta,\gamma$. Setting $\beta=\gamma=-\alpha \neq 0$ one finds
%
%\begin{equation}\label{}
%  2p|\gamma| \leq 0 ,
%\end{equation}
%which cannot be satisfied due to $|\gamma|>0,p>0$.
%\end{ex}
%Let us recall that any unital qubit channel can be realized as follows
%
%\begin{equation}\label{}
%  \Phi(\rho) = \sum_{\alpha=0}^3 p_\alpha U\sigma_\alpha V \rho V^\conj \sigma_\alpha U^\conj ,
%\end{equation}
%with unitary $U,V$ and probability distribution $p_\alpha$. It means that $\Phi$ is unitarily equivalent to Pauli channel. Hence, condition (\ref{MAIN}) provides necessary and sufficient conditions for the existence of such channel mapping $\rho_k$ to $\tau_k$.

\bibliographystyle{plainnat}
\bibliography{alberti_uhlmann_ver_Oct_01.bib}

\end{document}